\documentclass[aps,amsfonts,prd,nofootinbib,11pt]{article}

\usepackage{amsmath,amssymb}

\usepackage[numbers,sort&compress]{natbib}

\usepackage[utf8]{inputenc}

\usepackage{color}
\usepackage{graphicx}
\usepackage{wrapfig}
\usepackage{amsmath, amsthm, amssymb}

\usepackage{theoremref}

   \newtheorem{thm}{Theorem}[section]
   
   \newtheorem{lemma}[thm]{Lemma}

   \newtheorem{definition}{Definition}

\newcommand\defeq{\stackrel{def}{=}}

\usepackage{soul}
\usepackage{xcolor}

\usepackage{mwe}
\usepackage{float}
\usepackage[font=small]{caption}
\numberwithin{equation}{section}

\usepackage{fullpage}
\usepackage{authblk}

\author[1,2]{Yaron Oz}
\author[3]{Ittai Rubinstein}
\author[3]{Muli Safra}

\affil[1]{Raymond and Beverly Sackler School of Physics and Astronomy, Tel-Aviv University, Tel-Aviv 69978, Israel}
\affil[2]{School of Natural Sciences, Institute for Advanced Study, Princeton NJ, USA}
\affil[3]{Blavatnik School of Computer Science, Tel-Aviv University, Tel-Aviv 69978, Israel}

\date{\today}
\usepackage{theoremref}

\begin{document}

\title{Multivariate Generating Functions for Information Spread on Multi-Type Random Graphs}
\maketitle

\begin{abstract}

We study the spread of information on multi-type directed random graphs. In such graphs the vertices
are partitioned into distinct types (communities) that have different transmission rates between themselves and with other
types. 
We construct multivariate generating functions and use multi-type branching processes to derive 
an equation for the size of the large out-components in multi-type random graphs with a general class of degree distributions.
We use our methods to analyse the spread of epidemics and verify the results with population based
simulations.

\end{abstract}

\section{Introduction}

Random graphs \cite{Bol,Ran} provide a framework for studying the spread of information
in networks and
have been successfully applied to a wide range of dynamical systems such as 
social networks, epidemic spread and the internet \cite{New, review1, review2,web}.
In the simplest setup, one considers a complete graph where any two vertices are connected by a link
with a prescribed probability $p$, independently for each pair of vertices. 
We will refer to such a link between two vertices as "transmission", which means that some
information passed between the two vertices, and one often studies the fraction of the graph nodes to which information originating from some vertex is expected to reach.
In the language of epidemic spread, which we will be using, such nodes are called "infected".
Above a certain threshold value $p=p_c$ one has a giant connected component (GCC) or a large out-component
of infected nodes in the case of a directed graph. 

This type of graph (also known as Erd\"{o}s-Reyni graphs) is often generalized in one of two ways.
The first is by adding a more general degree distribution \cite{PRE_GCC, N2, corona1st, corona3rd, herd1, herd3, herd4, herd5, He}.
Since the existence of each edge in an Erd\"{o}s-Reyni graph is i.i.d, the in and out degrees of all vertices are distributed according to a Poisson distribution.
However, in many real-world applications the degrees of the vertices often exhibit a higher variance, resulting in different statistical behaviours of the resulting graphs.
Newman et al. \cite{PRE_GCC} presented a powerful mathematical tool using generating functions that allows one to analyse the spread of information on several simple graph structures with general degree distributions.
This tool works by analysing the generating function of the size of connected components and used to calculate the size of the greatest connected component (GCC) of the graph.
In \cite{N1} these tools were generalized to study percolations of several other families of graph.

Recently, the COVID-19 pandemic has spurred a lot of renewed interest in the study of the spread of diseases.
In particular, due to the heterogeneous spread of COVID-19, which is often characterized by super-spreading events and super-spreading individuals, a great deal of effort has gone into assessing the effects of heterogeneity in epidemiology.
Some of these papers focus on the early termination probability of an epidemic (the probability that the disease would disappear without intervention or herd immunity) when it is characterized by heterogeneous spread \cite{corona1st, althouse2020superspreading}.
Others focused on the effect of heterogeneity on the end of an epidemic when it is already widespread.
In \cite{herd1, herd3, herd4} the authors asked when we would reach herd immunity in the sense that the effective reproductive index $R_{\text{eff}}$ would be below $1$.
They use a variant of the SIR model to show that for negative binomial distributions, the percentage of the population infected before it reaches herd immunity can be characterized by a simple formula.

In \cite{He} the authors use a graph based model to estimate the total number of infected individuals until the epidemic is completely eliminated due to the herd immunity.
These models predict very different outcomes, due to the fact that at the point in time when $R_{\text{eff}} = 1$, there can still be many infected individuals causing an "after-burn" effect.
These models were compared and their limitations were explored in detail in \cite{corona3rd}.
Unlike most of the previously mentioned models, Nielsen et al. \cite{herd5} assume that heterogeneity is only in the infectiousness (and not in susceptibility).
In such a case, the heterogeneity is not expected to have a direct effect on the portion of the population reached by the disease, but can be used to model policies aimed at reducing the infectiousness of super-spreading individuals and super-spreading events.

In most of these examples, the graphs considered have a "single-type" structure, that is, all the vertices are equivalent.
This is characterized by the fact that while the in and out degrees of the vertices are not always the same, once these are selected, any outgoing edge from any vertex can be connected to any other vertex with the same probability (weighted by that vertex's in-degree).
In other words, for any vertex $v$ with in-degrees $i$, and for any vertex $u$ with out-degree $o$, the probability that the edge $e = (u, v)$ is in the graph is equal to $\Pr\left[e \in E\right] = cio$ for some constant $c$.
While being a good approximation in diverse cases, many real life networks are not complete graphs and are better approximated by a multi-type structure.
That is, the vertices are partitioned to different types (communities) that have a different interaction between themselves and with other types.

For instance, in \cite{wallinga2006using}, the authors use data on social interactions to estimate the rate at which members of different age groups are likely to infect one another with respiratory spread diseases.
Indeed, they found that age cohorts have a significant influence on the network of interactions.
For example, they found that most members of each age cohort are more likely to interact with members of their own cohort and that certain pairs of age cohorts interact more often than others.
Additionally, some age cohorts are in general more active than others, and as a result we would expect an epidemic to affect them with differing degrees.
Effectively modeling this phenomena could have a huge impact on our understanding of epidemics.

Multi-type graphs are most easily analysed when the in and out degrees are i.i.d Poisson distributions.
This model was proposed by Ball and Clancy \cite{original} and the size of the large out-components was derived.
In \cite{book}, this formula is presented with additional background (see Chapter 6).
In \cite{Minzer_2021} a simple condition for the existence of a giant connected connected component was proven.
However, these models do not take into account the heterogeneity which was also shown to have a significant impact on the epidemiological spread.
One possible combination of these models was considered by Britton et al. \cite{science}.
In this model one attributes to each individual a type $1 \leq i \leq r$ and an activity level $\alpha \in \mathbb{R}$, and the probability for an infection between any two individuals $u, v$ of types $i,j$ respectively is $\Pr\left[(u,v) \in E\right] \propto M_{i,j} \alpha(v) \alpha(u) $, where $M$ is the interaction matrix.

For the COVID-19 pandemic, it is commonly assumed that the activity levels are distributed according to a Gamma distribution with a small value of $k \approx 0.1$ \cite{herd1, herd3, herd4, herd5, He, corona1st, corona3rd}, and there is experimental evidence supporting this assumption. 
For instance, Lloyd-Smith et al. who estimates $k\approx 0.16$ from experimental data on SARS-1 outbreaks \cite{herd6}, and \cite{Adi} who showed that COVID-19 infections have a higher variance (and therefore have a lower $k$).
However, the methods used in \cite{science}, are only applicable to discrete distributions with a small support and as a result they cannot be easily used to model Poisson activity levels.
Instead, they consider a heuristic distribution where each type has three activity levels representing $25\%, 50\%, 25\%$ of the population with activity levels of $\frac{1}{2}, 1, 2$ respectively.

The aim of this paper is to adapt the tools developed in \cite{PRE_GCC} to the multi-type setting considered in \cite{Minzer_2021} by introducing a tool to model and calculate the size of the large out-components of a more general class of multi-type graphs, which includes Poisson graphs with an arbitrary activity level distribution.
This will allow us to use the activity level multi-type model introduced in \cite{science}, with activity levels distributed according to a Gamma distribution.
Our analysis will follow in the same general direction as the one in \cite{PRE_GCC}, adapting many of its intermediate results to the multi-type setting as well.

In our work we focus on directed graphs as they are more relevant for most of our use-cases.
For such graphs one needs to refine the definition of the GCC.
One possible approach, as taken in \cite{PRE_GCC}, is to look at greatest strongly connected component, i.e. the bow-tie model.
However, keeping for instance with the example of the spread of a disease, one is interested in the percentage of the population infected when the disease starts from a small fraction of the population (i.e. in the out-component of those initially infected).
We formalise the concept for general graphs and 
apply our framework in several settings, including a comparison of different activity levels in the multi-type setting considered in \cite{science}, and discuss its limitations.

The paper is organized as follows:
In Section 2, we define the main concepts of our framework and give an overview of the limitations and assumptions of our model.
In Section 3, we show how these can be combined to find the size of the large out-components.
In Section 4, we consider several examples of multi-type graphs and analyse them using the tools developed in Section 3.
In Section 5, we apply our results to modeling the spread of epidemics.
In Section 6, we discuss in detail the limitations of our framework.
Section 7 is dedicated to a discussion of the results and their implication.
In the appendices we outline in detail some technical aspects.

\section{The Model and Generating Functions}

In this section we lay the groundwork for the construction of the generating functions of multi-type directed random graphs.
Let $G = (V,E)$ be a random graph representing a population of $N=\left\vert V \right\vert$ individuals.
We partition $V$ into $r$ disjoint sets $V = V_1 \sqcup \cdots \sqcup V_r$ and for any vertex $u$ of type $i$ (i.e. $u\in V_i$) we define its $r$ dimensional in / out degree vectors to be

\begin{equation}
    \deg_{\text{in}, j} (u) = \left\vert \left\{ v\in V_j \mid (v, u) \in E \right\} \right\vert    \ ,
\end{equation}
and
\begin{equation}
   \deg_{\text{out}, j} (u) = \left\vert \left\{ v\in V_j \mid (u, v) \in E \right\} \right\vert  \ .
\end{equation}

We will say that $G$ is an $r$-type graph if it is created by a process, which, after assigning to each vertex a type (with proportions $\pi_i = \frac{n_i}{N}$) and in and out degree vectors, connects any two edges of the same type with the same probability.
Formally, we require that

\begin{equation}
    \forall i, j, u\in V_i, v\in V_j \;\;\;\; \Pr\left[ \left( u, v\right) \in E \right] = \frac{\deg_{\text{out}, j} (u) \deg_{\text{in}, i} (v)}{\sum\limits_{w \in V_j} \deg_{\text{in}, i} (w) }  \ .
\end{equation}

Furthermore, for all $i$, we will denote by $\left(I_{i,1}, \ldots, I_{i, r}\right)$, $\left(O_{i,1}, \ldots, O_{i, r}\right)$ the random variables corresponding to the in and out degree vectors of a random vertex of type $i$.
In other words:

\begin{equation}
    \begin{aligned}
        \forall \Vec{i}, \Vec{o} \in \mathbb{N} ^ r \;\;\;\; &\Pr\left[ \left(I_{i,1}, \ldots, I_{i, r}\right) = \Vec{i} \wedge \left(O_{i,1}, \ldots, O_{i, r}\right) = \Vec{o} \right] = \\
        &\Pr\left[ \left( \left( \deg_{\text{in}, 1} (u), \ldots, \deg_{\text{in}, r} (u) \right) = \Vec{i} \right) \wedge \left( \left( \deg_{\text{out}, 1} (u), \ldots, \deg_{\text{out}, r} (u) \right) = \Vec{o} \right) \mid  u \in V_i \right]
    \end{aligned} 
\end{equation}

Throughout this manuscript, we will assume that the second moments of $I_{i,j}$ and $O_{i,j}$ are finite for all $i,j$.
This condition is assumed both in our definition of sufficient independence (see Section \ref{subsec:suff_ind}) and for our analysis (see Section \ref{subsec:additive_structure}).
This assumption holds for the class of graphs that will be the main focus of our applications.
However, it does limit our model in some aspects and we consider this limitation in-depth in Section \ref{subsec:stds}.

\subsection{Neighbors and the Rate of Expansion}
\label{subsec:neighbor_defs}

Let $v \in V$ be a vertex in $G = (V, E)$.
We will define the $d$th neighbors of $v$, denoted by $\mathcal{N}_d (v)$, to be the set of vertices that can be reached from $v$ in exactly $d$ steps.
Furthermore, we will define $N_{i,j} (d)$ to be the random variable obtained by choosing a random vertex $v$ of type $i$, and counting the number of $d$th neighbors of type $j$ it has.
Thus, 

\begin{equation}
    \forall i,j,d, \Vec{\nu} \in \mathbb{N} ^ r \;\;\;\; \Pr \left[ \left( N_{i, 1} (d), \ldots, N_{i, r} (d) \right) = \Vec{\nu} \right] = \Pr \left[ \left( \left\vert \mathcal{N}_d (v) \cap V_1 \right\vert, \ldots, \left\vert \mathcal{N}_d (v) \cap V_r \right\vert \right) = \Vec{\nu} \mid v \in V_i \right]
\end{equation}

Using this definition, we extend the definition of a reproductive index $R_{\text{eff}}$ from the study of simple single type graphs (on which the spread of information or diseases can be characterised by the SIR equation), to the {\em rate of expansion} of general families of random graphs.
We will base our definition on an experimental method commonly used to approximate the value of the reproductive index in  epidemics (see e.g. \cite{R01, R02, R03, R04}).
This method works by measuring the number of infected individuals as a function of time $I(t)$ and fitting it to an exponential growth $I(t) \approx c {R_{\text{eff}}} ^ {d(t)}$, where $d(t) = \frac{t}{\tau}$ is assumed to be the depth of the infection graph as a function of time and $\tau$ is the time it takes a newly infected individual to infect others.

In general, the expected number of infected nodes at depth $d$ reads:
\begin{equation}
N_{\text{ave}} (d) = \mathop{\mathbb{E}}\limits_{v \in V} \left\vert \mathcal{N}_d (v) \right\vert = \sum_{i, j} \pi_i N_{i,j} (d) \ .
\label{av}
\end{equation}
Following the methods of \cite{R01, R02, R03, R04}, we fit $N_{\text{ave}} (d)$ to an exponential growth and define:

\begin{equation}
R_{\text{eff}} \defeq \lim_{d\rightarrow \infty} \left( N_{\text{ave}}(d) \right) ^ {\frac{1}{d}}  =  \lim_{d\rightarrow \infty} \left( \mathbb{E} \sum_{i,j} \pi_i N(d)_{i,j} \right) ^ {\frac{1}{d}}
\ .
\label{R0}
\end{equation}

In order to help clarify these definitions we will consider a simple example.
In Section \ref{subsec:neighbor_analysis}, we will show how the $R_{\text{eff}}$ parameter can be computed for general degree distributions (assuming sufficient independence).

\subsubsection{A Simple Example}

We will consider a simple case where both the in and out degrees of any vertex are distributed according to a delta distribution.
Let $G = (V,E)$ be a $r=2$-type graph, with an evenly divided population ($\pi_1 = \pi_2 = \frac{1}{2}$), where all the vertices of type $1$ have in-degree $(1, 1)$ and out-degree $(1, 2)$, and all the vertices of type $2$ have in-degree $(2, 0)$ and out-degree $(1, 0)$.

In other words, every vertex of type $1$ has $2$ incoming edges ($1$ from another vertex of type $1$ and $1$ from a vertex of type $2$), and $3$ outgoing edges ($1$ to another vertex of type $1$ and $2$ to vertices of type $2$).
Furthermore, every vertex of type $2$ has $2$ incoming edges (both from vertices of type $1$) and $1$ outgoing edge (to a vertex of type $1$).
In Figure \ref{fig:example_graph} we present a visual representation of such a graph with $N=10$ vertices.

Let $v$ be some vertex.
Denote the number of $d$th neighbors of type $1, 2$ of $v$ by $N_1 (d), N_2 (d)$ respectively. 
Because of our assumptions, these can be used to compute the number of $d+1$th neighbors of $v$ (so long as there are no repeat neighbors - see Section \ref{subsec:additive_structure}).
\begin{equation}
    \left(\begin{matrix}
                    N_1 (d + 1) \\ 
                    N_2 (d + 1)
    \end{matrix}\right) =
    \left(\begin{matrix} 
                    1   &   1\\
                    2   &   0
    \end{matrix}\right)
    \left(\begin{matrix} 
                    N_1 (d) \\ 
                    N_2 (d) 
    \end{matrix}\right) \ .
    \label{eq:recursion}
\end{equation}

This forms a linear recurrence rule for $N_{i,j} (d) = N_{i,j} (d-1) + 2 N_{i,j} (d-2)$, the solutions of which can be shown to be of the form $N_{i,j} (d) = a_{i,j} \lambda_1 ^ d + b_{i,j} \lambda_2 ^ d$, where $\lambda_{1,2} = -1, 2$ are the roots of $f(x) = x^2 - x - 2$ (and also the eigenvalues of the matrix in equation \eqref{eq:recursion}) and $a_{i,j}, b_{i,j}$ are listed below.

\begin{equation}
    a_{i,j} =
    \frac{1}{5} \left(\begin{matrix} 
                    -1   &   2\\
                    2   &   -4
    \end{matrix}\right)
    \;\;\;\;\;\;\; 
    b_{i,j} = \frac{1}{2} \left(\begin{matrix} 
                    1   &   1\\
                    1   &   1
    \end{matrix}\right) \ .
\end{equation}

Therefore, $N_{\text{ave}} (d) = \sum_{i,j} \pi_j N_{i,j} (d) = a \lambda_1 ^ d + b \lambda_2 ^ d$, where $a = -\frac{1}{5}$ and $b = \frac{1}{2}$.
Clearly the second term dominates and $R_{\text{eff}} = \lim_{d\rightarrow \infty} \left(N_{\text{ave}} \left(d\right)\right) ^ {\frac{1}{d}} = \lambda_2 = 2$.

In Section \ref{subsec:neighbor_analysis} we will show that in general the value of $R_{\text{eff}}$ is determined by the largest eigenvalue of a "weighted interaction matrix".
In this example the interaction matrix would be $ M = \left(\begin{matrix} 1   &   1\\2   &   0\end{matrix}\right)$ and its largest eigenvalue is indeed $2$.

\begin{figure}[H]
\centering
\includegraphics[width=0.7\columnwidth]{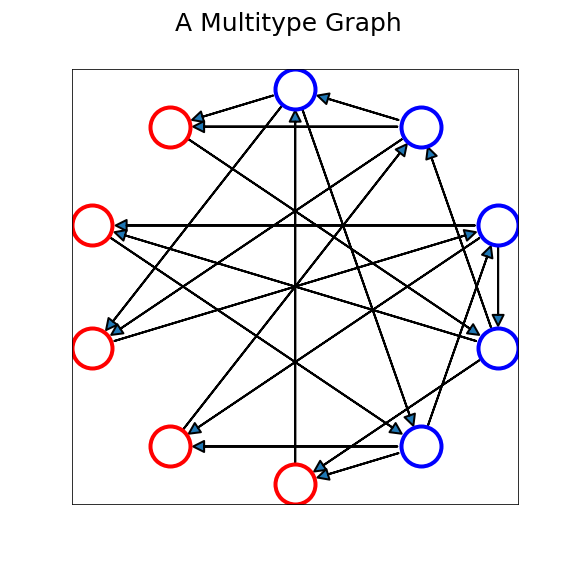} 
\caption{
An example of a $2$-type graph with a delta degree-distribution. This graph has $r=2$ types of vertices (the red and blue circles). Each vertex of type $1$ ({\color{blue} blue}) has in-degree $(1, 1)$ and out-degree $(1, 2)$, and each vertex of type $2$ ({\color{red} red}) has in-degree $(2, 0)$ and out-degree $(1, 0)$.
}
\label{fig:example_graph}
\end{figure}

\subsection{Large Out-Components}

We would like the large out-component of our graph to model the spread of an information such as a disease or some other transmission, assuming that it starts from a small fraction of the graph.
In other words, we require this component to include those vertices that are likely to be reached when spanning from a small fraction of the vertices when chosen at random.

\begin{definition}
    We say that a vertex $v$ is a \emph{member of the large out-component} of a graph $G$ if the set of vertices spanned by $v$ on the reversed graph (which we call its in-component) is very large.
\end{definition}

As we will see in the next sections, in random graphs, a constant fraction of the vertices would have bounded in-components, while the rest would have very large in-components.
If we look at the set spanned by a small fraction of the vertices in the graph, any vertex with a bounded in-component is unlikely to be a member of this spanned set, while a vertex with an unbounded in-component, is unlikely not to have at least one of these starting vertices in its in-component.
Therefore, in order to understand the fraction of the graph "easily" spanned, we will need to consider the fraction of the vertices in the graph with large in-degrees.

\subsection{Sufficient Independence}
\label{subsec:suff_ind}

A major step in the analysis of \cite{PRE_GCC} is to consider the distribution of the in-degrees of vertices when weighted by their out-degrees, i.e. giving vertices with higher out-degrees a larger weight when averaging.
Generalizing this to the multi-type case is not straightforward, since it is not clear how one should take into account the $r$-dimensional multi-type weight.
One possible approach is to take the total out-degree, but this can be unreliable.
Consider, for instance, the case where one of the types has only incoming edges, and no outgoing ones. 
Clearly, edges to this type are "dead ends" and should not be taken into account.
However weighting vertices by their total out-degree does take these edges into account, and can lead us to use an incorrect distribution on the vertices.

In order to avoid this issue, we will restrict the analysis to cases where it does not matter by which dimension of the out-degree we weight a vertex.
We will call such graphs "sufficiently independent" (see formal definition below).
Solving for this regime will still allow us to accomplish our main goal of combining the theory of multi-type graphs presented in \cite{original, book, Minzer_2021}, with the theory of graphs that have more complex degree distributions and correlations between in and out degrees \cite{corona3rd, herd1, herd3, herd4, herd5, herd6, He, science}.

We will do this using a model similar to the one considered by Britton et al. in \cite{science}.
As discussed above, the authors attribute to each individual a type and an activity level. 
The type of an individual corresponds to their age cohort (controlling which other individuals they are likely to interact with), and the activity level corresponds to their overall propensity to infect and to become infected.
We will show that if each individual has a single activity parameter affecting their interaction with all of the other types in the same manner (as is the case in \cite{science}), then the graph is sufficiently independent.
This allows us to model the behaviour of multi-type population with general activity level distributions (whereas the model proposed by Britton et al. only covers discrete activity distributions).
Additionally, as will be clear from the definition, in cases where each type interacts with only one other type, the graph is always sufficiently independent.
Therefore, bipartite graphs, which are often considered when modeling the spread of sexually transmitted diseases \cite{N2}, are also covered by this definition.

That being said, there may be many physical systems which are not sufficiently independent
and it is possible to generalize our analysis to cover such cases.
However, this would result in a rather complex model which is not likely to be applicable for most real world use cases where one does not have a sufficiently detailed distribution to begin with.
In Section \ref{subsec:lim_suff_ind}, we will further explore this limitation by showing an example of a graph for which sufficient independence does not hold and remark on how the methods shown in this paper might be adapted to that regime.

Formally, we will say that a distribution of in/out-degrees, i.e. a distribution of pairs of vectors in $\mathbb{N}^r$, is sufficiently independent if the distributions of the out-degree of a vertex is the same when weighted by any coordinate of its in-degree. Similarly, the in-degrees can be weighted by any of the out-degrees.
Consider a random vertex $v$ of type $i$.
Denote by $I_{i,j}$ the distribution of its $j$th in-degree, i.e. the number of edges from vertices of type $j$ to $v$, and by $O_{i, j}$ the distribution of the $j$th out-degree of $v$, i.e. the number of edges going from $v$ to vertices of type $j$. We denote $\vec{I_i} = (I_{i,1},...,I_{i,r}), \vec{O_i} = (O_{i,1},...,O_{i,r})$.
Keeping in mind that $\vec{I_i}$ and $\vec{O_i}$ can be correlated as they are the in and out degrees of the same vertex, our condition for sufficient independence is that for all $i,k,l$ (such that $\langle O_{i,k} \rangle,  \langle O_{i,l} \rangle \neq 0$ for equation \ref{eq:in1} or $\langle I_{i,k} \rangle, \langle I_{i,l} \rangle \neq 0$ for equation \ref{eq:in2}) and for all $\vec{x}$:
\begin{equation}
    \begin{aligned}
\sum_y \frac{y}{\langle O_{i,k} \rangle} \Pr\left[\vec{I_{i}} = \vec{x} \wedge O_{i,k} = y\right] = \sum_z \frac{z}{\langle O_{i,l} \rangle} \Pr\left[ \vec{I_{i}} = \vec{x} \wedge O_{i,l} = z\right] \ ,
    \end{aligned}
     \label{eq:in1}
\end{equation}

\begin{equation}
    \begin{aligned}
\sum_y \frac{y}{\langle I_{i,k} \rangle} \Pr\left[ \vec{O_{i}} = \vec{x} \wedge I_{i,k} = y\right] = \sum_z \frac{z}{\langle I_{i,l} \rangle} \Pr\left[ \vec{O_{i}} = \vec{x} \wedge I_{i,l} = z\right]  \ ,
    \end{aligned}
    \label{eq:in2}
\end{equation}
where $\langle...\rangle$ is the mean value of the distribution.

Under this assumption, we can define the in/out degree distribution weighted by the out/in degree distribution.
\begin{definition}
\label{def:biased}
Let $G$ be a sufficiently independent multi-type random graph with in and out degree distributions $I_{i,j}, O_{i,j}$, respectively. We define its biased in/out degree distributions as:
\begin{equation}
    \begin{aligned}
\Pr[B\vec{I_{i}} = \vec{x}] = \sum_y \frac{y}{\langle O_{i,k} \rangle} \Pr\left[\vec{I_{i}} = \vec{x} \wedge O_{i,k} = y\right] \ ,
    \end{aligned}
\end{equation}
\begin{equation}
    \begin{aligned}
\Pr[B\vec{O_{i}} = \vec{x}] = \sum_y \frac{y}{\langle I_{i,k} \rangle} \Pr\left[\vec{O_{i}} = \vec{x} \wedge I_{i,k} = y\right]  \ .
    \end{aligned}
\end{equation}
\end{definition}
These are well defined even though we did not specify $k$ precisely because of our independence assumption
(\ref{eq:in1}) and (\ref{eq:in2}).

\subsection{Multivariate Generating Functions}
In order to generalize the results of \cite{PRE_GCC} to the multi-type setting, we will first generalize their analysis of single variable generating functions to the multivariate regime.
Working with multivariate generating functions instead of high dimensional probability distributions will allow us to complete our analysis with a greater deal of elegance.

Consider first a random graph with a large number $N$ of single type of vertices. 
Let $A$ be a random variable that corresponds to the distribution of the degrees of these vertices.
Define a generating function by \cite{PRE_GCC}:
\begin{equation}
G_A (z) = \sum_{\alpha = 0}^{\infty} 
p_{\alpha} {z}^{\alpha} \ ,
\label{one}
\end{equation}
where $p_{\alpha} = \Pr [A=\alpha]$ is the probability that a vertex chosen at
random will have a degree $\alpha$. The normalization of the sum of the probabilities to one gives 
$G_A (1)=1$. For this reason as well (\ref{one}) is well defined in the range $|z| \leq 1$.

The generating function encodes the information of the probability distribution. 
The probabilities $p_{\alpha}$ are obtained by appropriate derivatives of (\ref{one})
with respect to $z$ at $z=0$:
\begin{equation}
p_{\alpha} = \frac{1}{\alpha !}\frac{d^{\alpha} G_A(z)}{d z^{\alpha}}\vert_{z=0} \ ,   
\end{equation}
while the moments of the distribution are obtained by the derivatives of (\ref{one}) at $z=1$:
\begin{equation}
\langle \alpha^n \rangle = \sum_{\alpha = 0}^{\infty} \alpha^n p_{\alpha} = \left(z \frac{d}{d z}\right)^{\alpha} G_A(z)\vert_{z=1} \ .   
\end{equation}
In the following we generalize (\ref{one}) in two ways, first by having multi-type
vertices making $z$ a vector and second by allowing multiple random variables
by making $A$ a vector.
\begin{definition} [Multivariate Generating Functions]
Let $A:\Omega\rightarrow \mathbb{N}^r$ be a random variable whose values are non-negative integer values vectors of dimension $r$ (where $r$ is the number of types in our graph).
Its multivariate generating function is: 
\begin{equation}
G_A \left(\Vec{z}\right) = G_A \left(z_1 , \dots, z_r\right) \defeq \sum_{\alpha_1,...,\alpha_r} p_{\alpha_1...\alpha_r} {z_1}^{\alpha_1} \cdot \dots \cdot 
{z_r}^{\alpha_r}  \ ,
\end{equation}
where $p_{\alpha_1..\alpha_r} = {\Pr [A=\Vec{\alpha}]}$ is the probability that $A$ has the value
$\Vec{\alpha} = (\alpha_1,...,\alpha_r)$.
Let $\Vec{A} = \left(A_1 , \dots, A_m\right)$ be a vector of such random variables, $A_i :\Omega\rightarrow \mathbb{N}^r$.
We define their $r\rightarrow m$ multivariate generating function as:
\begin{equation}
G_{\Vec{A}} \left(\Vec{z}\right) \defeq \left(G_{A_1} \left(\Vec{z}\right), \dots G_{A_m} \left(\Vec{z}\right)\right) \ ,
\end{equation}

where
\begin{equation}
G_{A_i} \left(\Vec{z}\right) = G_{A_i} \left(z_1 , \dots, z_r\right) \defeq \sum_{\alpha_1,...,\alpha_r} p_{\alpha_1...\alpha_r}^{(i)} {z_1}^{\alpha_1} \cdot \dots \cdot 
{z_r}^{\alpha_r}  \ ,
\end{equation}
with $p_{\alpha_1..\alpha_r}^{(i)} = {\Pr [A_i=\Vec{\alpha}]}$.
\end{definition}

The next ingredient in the analysis is the product of the generating functions of independent variables. 
Consider the products of the generating function in the single-type case
\cite{PRE_GCC}:
\begin{equation}
\prod_{i=1}^{k} G_{A_i} (z)  = \sum_{\alpha_1+...+\alpha_k = n, n=0}^{\infty} 
p_{\alpha_1}^{(1)}\cdot\cdot\cdot  p_{\alpha_k}^{(k)}{z}^{(\alpha_1+...+\alpha_k)} = G_{\sum_{i=1}^{k} A_i} (z)  \ ,
\label{power}
\end{equation}
where $p_{\alpha_i}^{(i)} = \Pr [A_i=\alpha_i]$ and
the summation is over all the $k$ partitions of $n$.
We generalize this in the multi-type case as follows.
\begin{lemma}
Let $\Vec{A}, \Vec{B}$ be independent vectors of random variables, i.e. $A_i$ and $B_i$ are independent for all $i$, with $r \rightarrow m$ multivariate generating functions $G_{\Vec{A}}, G_{\Vec{B}}$.
Let $\Vec{C} = \Vec{A} + \Vec{B}$ be their sum, $C_i = A_i + B_i$.
Then
$G_{\Vec{C}} = G_{\Vec{A}} \odot G_{\Vec{B}}$, where the $\odot$ denotes point-wise multiplication.
\end{lemma}
\begin{proof}
As in the single-type case (\ref{power}) this lemma  can be shown by writing the coefficients of $G_{\Vec{A}} \odot G_{\Vec{B}}$ directly and seeing that they correspond to a convolution of the coefficients of $G_{\Vec{A}}$ and $ G_{\Vec{B}}$.
\end{proof}

Consider next the composition of generating functions.
In the single type case one has \cite{PRE_GCC}:
\begin{equation}
    G_A(G_B(z)) = \sum_{\alpha= 0}^{\infty} 
p_{\alpha} (G_B(z))^{\alpha} = \sum_{\alpha= 0}^{\infty} 
p_{\alpha} (G_{B_1 + \dots + B_\alpha}(z))\ .
\label{compo}
\end{equation}
Equation (\ref{compo}) means that applying the generating function of $A$ to the generating function of $B$ is equivalent to producing $A$ independent copies of the distribution $B$ and summing over them.
This can be generalised to the multivariate case in the following manner.

\begin{definition}
Let $\Vec{A}, \Vec{B}$ be vectors of random variables with $n \rightarrow m$ and $m \rightarrow k$ multivariate generating functions $G_{\Vec{A}}, G_{\Vec{B}}$ respectively.
We define their composition:
\begin{equation}
G_{\Vec{B}}\left(G_{\Vec{A}} \left( \Vec{z} \right) \right) \defeq \left( G_{B_1} \left( G_{\Vec{A}} \left(\Vec{z}\right) \right), \dots G_{B_k} \left( G_{\Vec{A}} \left(\Vec{z}\right) \right) \right) \ .
\end{equation}
\end{definition}

\begin{lemma} [Composition Lemma]
    Let $\Vec{A}, \Vec{B}$ be independent vectors of random variables (i.e. $A_i$ and $B_i$ are independent for all $i$), with $r \rightarrow m$ and $m \rightarrow k$ multivariate generating functions $G_{\Vec{A}}, G_{\Vec{B}}$ respectively.
    For each $i \in \{1, \dots, r\}$, we define the random variable $C_i$ to be the result of the following random process:
    \begin{itemize}
        \item A vector $\Vec{b}\in\mathbb{N}^m$ is chosen according to the distribution $B_i$.
        \item For each $j \in \{1, \dots, m\}$, $b_j$ vectors $\Vec{a}^{j}_1, \dots, \Vec{a}^j_{b_j} \in \mathbb{N} ^ r$ are chosen i.i.d according to the distribution $A_j$.
        \item The result $C_i$ is the sum over all of these $C_i = \sum_{j \in \{1, \dots, m\}} {\sum_{\beta \in \{1, \dots, b_j\}} {a^j_\beta}}$
    \end{itemize}
    
    Then,
    \begin{equation}
    G_{\Vec{C}}\left( \Vec{z} \right) = G_{\Vec{B}}\left(G_{\Vec{A}} \left( \Vec{z} \right) \right) \ .
    \label{gen}
    \end{equation}
    Equation (\ref{gen}) is a generalization of (\ref{compo}).
    \label{lem:comp}
\end{lemma}
\begin{proof}
For each value of $\Vec{b}$ the generating function of the conditional distribution is
\begin{equation}
G_{C_i \mid B_i = \Vec{b}} \left( \Vec{z} \right) =\prod_{j \in \{1, \dots, m\}} {\left(G_{A_j} \left( \Vec{z} \right)\right) ^ {b_j}} \ ,
\end{equation}
and the final generating function of $C_i$ is

\begin{equation}
    \begin{aligned}
G_{C_i \mid B_i = \Vec{b}} \left( \Vec{z} \right) &= \sum_{\Vec{b} \in \mathbb{N} ^ m} \Pr[B_i = \Vec{b}] G_{C_i \mid B_i = \Vec{b}} \left( \Vec{z} \right) =\\
&=\sum_{\Vec{b} \in \mathbb{N} ^ m} \Pr[B_i =\Vec{b}] \prod_{j \in \{1, \dots, m\}} {\left(G_{A_j} \left( \Vec{z} \right)\right) ^ {b_j}} =  G_{B_i}\left(G_{\Vec{A}} \left( \Vec{z} \right) \right) \ .
    \end{aligned}
\end{equation}
\end{proof}

\section{Multi-Type Branching Process}
In this section we generalize the analysis of  \cite{PRE_GCC} from single-type random
graphs to multi-type random graphs \cite{Minzer_2021}.
We use the multi-type Galton-Watson branching process and derive the equations for
the sizes of the large out-components.
We show that this is indeed well defined, since a certain fraction of the vertices in the graph will have a large in-component, meaning that any large enough set will span them with high probability.

\subsection{The Additive Structure of Connected Components}
\label{subsec:additive_structure}

Following the same steps as the derivation in \cite{PRE_GCC}, we will show that if the size of the connected component originating from some vertex $v$ is bounded, then the probability that it contains a cycle is $O\left(N^{-1}\right)$.
Moreover, we will show so long as the set of vertices that can be reached from $v$ through $d$ steps (where $d$ is not necessarily $O(1)$) is only a negligible fraction of the graph, then it is "probably almost a tree" (wp $1-o(1)$ the number of vertices is at least $1 - o(1)$ times the number of edges).

Proving this would show an additive structure to the size of the sets spanned by vertices in the graph.
Let $v$ be a vertex in the graph, and define the $d$th neighbors of $v$, $\mathcal{N}_d (v)$, to be the set of vertices that can be reached from $v$ in exactly $d$ steps.
If the sub-graph containing them is tree-like, then 
\begin{equation}
\left\vert \mathcal{N}_d (v) \right\vert = \left\vert \bigcup\limits_{w \in \mathcal{N}_1 (v)} \mathcal{N}_{d-1} (w) \right\vert \approx \sum\limits_{w \in \mathcal{N}_1 (v)} \left\vert\mathcal{N}_{d-1} (w) \right\vert
\ .
\end{equation}
This additive structure in the components will be the basis of most of our analysis.

Let $v$ be a vertex in the graph.
Consider a BFS (breadth first search) walk on the vertices spanned by $v$ (if necessary, abort the walk when it reaches depth $d+1$), and let $u$ be a vertex in this traversal.
This walk begins at a starting vertex $v$, and enumerates the vertices spanned by it, in order of their distance from $v$.

In the step of the BFS process starting from $u$, the outgoing edges $e=(u,t)$ from $u$ are enumerated and for each $e$, we view its target $t$ to see if it has already been visited in this BFS process.
We will prove our claims by bounding the probability that the target of any step is already in the spanned set.
By our construction of the graph, the probability of any given vertex $w$ to the spanned set is determined only by its type and its in-degree.
In other words, any given $w \in \text{span}\{v\}$ might be a more likely candidate for being the target of an outgoing edge from $u$, either due to its high in-degree ($w$ is selected from a distribution that is somewhat biased towards a higher in-degree) or because of its type, but by our construction, it cannot have any other bias towards it.
We will show that neither of these can skew the distribution too much.

Since $\pi_j$ fraction of the population are of type $j$, conditioning on the fact that $w$ is of type $j$ can only increase the probability of $w$ appearing again by a factor of $\pi_j ^{-1}$.
Furthermore, assuming $u$ is of type $i$ and $w$ is of type $j$, $w$ is only $\frac{\deg_{\text{in}, i} (w)}{\mathbb{E} I_{i,j}}$ more likely to be selected due to its in-degree.
In other words:

\begin{equation}
\Pr\left[ t=w\right] \leq \frac{ 1 }{N } \cdot \frac{\deg_{\text{in},i} \left(w\right)}{ \mathbb{E} I_{i,j} \pi_j}\ .
\label{eq:col_prob}
\end{equation}

In principle, it is possible for some vertices to have a very large in/out degree.
For instance, a common model for the spread of sexually transmitted diseases assumes a power-law degree distribution with power $\alpha \approx 3.2$ (see Section \ref{subsec:stds}).
In such a case we would expect to see some individuals with out-degrees of the order of $N^{\alpha^{-1}} \gg 1$.
However, so long as the second moments of $I_{i,j}$ are finite, this becomes the exception rather than the rule.

Denote by $\mathcal{S} \subseteq \text{span}\{v\}$ the set of vertices already visited by the BFS.
Averaging over $w \in \mathcal{S}$ of type $j$, their $i$th in-degree would be $\mu_{i,j} = \frac{\mathbb{E} {I_{i,j}} ^ 2}{\mathbb{E} I_{i,j}}$ .
Since these are all finite, we can insert them into the Equation \ref{eq:col_prob} to get:

\begin{equation}
\Pr\left[ t \in \mathcal{S} \right] \leq \frac{ \left\vert\mathcal{S}\right\vert }{N } \cdot \max_{i,j} \left(\frac{\mathbb{E} \left({I_{i,j}} ^ 2\right)}{ \left(\mathbb{E} I_{i,j}\right)^2 \pi_j}\right) \leq \frac{ \left\vert \text{span}\{v\} \right\vert }{N } \cdot \max_{i,j} \left(\frac{\mathbb{E} \left({I_{i,j}} ^ 2\right)}{ \left(\mathbb{E} I_{i,j}\right)^2 \pi_j}\right) = O\left(\frac{ \left\vert \text{span}\{v\} \right\vert }{N }\right)
\ .
\end{equation}
This traversal has enumerated over all of the edges spanning from $v$ and we have shown that on average, almost all of them lead to previously unexplored vertices.
Therefore, using the Markov inequality, we can conclude that with high probability almost all of them lead to new vertices.
In other words, the size of the connected component of a vertex is almost always close to the sum of the sizes of the connected components of its neighbors.

\subsection{Multi-Type Galton-Watson Analysis}

Consider an infection originating from a single individual.
If, on average, an infected individual  tends to infect less than one new individual, it is clearly unlikely that this infection would reach a significant portion of the graph (for instance, since the expectancy of infected individuals is less than $1$, allowing us to bound this probability by Markov's inequality).
Even if on average an infected individual tends to infect more than one new individual, the infection can still peter out early on (for instance if the patient zero happens not to be very infectious even though the average patient is infectious).
However, once sufficiently many individuals are infected, it becomes extremely unlikely for the disease not to grow (due to the law of large numbers).
This growth will not go completely unchecked, and will end once a large enough percentage of the population has been infected and herd immunity is reached.
If the first moment of its biased degree distribution (see definition \ref{def:biased}) is finite, then this happens after some constant fraction of the population has been infected.

Our goal in what we call a "Galton-Watson analysis" is to find the probability that a randomly chosen vertex will be among those who span a large portion of the graph.
Analysing this probability will help us find the portion of the population expected to be in each of these large out-components.
Naturally, this analysis depends on the distribution of the out-degrees of the vertices, but it can also depend on the joint in and out degree distribution.
As we will show, when the in and out degree distributions are sufficiently independent, it suffices to know the out-degree distribution and the weighted out-degree distribution in order to find this probability.

Let $O_i = \left(O_{i,1}, \ldots, O_{i,r}\right)$ denote the distribution of the out-degree of an individual of type $i$ when chosen uniformly at random. 
Let $BO_i$, the biased out-degree distribution, be the out-degree distribution of an individual of type $i$
when the individual is chosen weighted by his in-degree (see definition \ref{def:biased}).
Denote $\Vec{O} = (O_1,...,O_r)$,  $B\Vec{O} = (BO_1,...,BO_r)$ and  the corresponding multivariate generating functions by $G_{\Vec{O}}$ and $G_{B\Vec{O}}$.

\subsubsection{Neighbors Distribution}

\label{subsec:neighbor_analysis}

Consider the random variables $N(d)_{i,j}$ defined in Section \ref{subsec:neighbor_defs} that denote the distribution of $d$th neighbors of type $j$ of a random vertex of type $i$.
Let $BN(d)_{i,j}$ denote the same distribution over vertices weighted by their in-degree.
Since the number of $d$th neighbors of a vertex depend only on the number and type of its out-going edges, these are well-defined for sufficiently independent graphs.
For instance,
\begin{equation}
N(1)_{i,j} = O_{i,j},~~~BN(1)_{i,j} = BO_{i,j} \ .
\end{equation}

Next we consider the case where $d=2$.
Neglecting vertices that can be reached from our starting vertex in more than one way (see Section \ref{subsec:additive_structure}), the distribution of the number of second neighbors is the distribution of the sum of the number of first neighbors of each of the first neighbors of our starting vertex.
Since the neighbors are independent of each other, we can apply the composition lemma (\ref{lem:comp}).
Therefore, the generating function of the distribution of the number of second degree neighbors of a random vertex reads: 
\begin{equation}
G_{\Vec{N}(2)} \left( \Vec{z} \right) = G_{\Vec{O}} \left(G_{B\Vec{O}} \left( \Vec{z} \right) \right) \ .
\end{equation}
In the general $d$th neighbors case, we have:
\begin{equation}
G_{B\Vec{N}(d)} \left( \Vec{z} \right) = G_{B\Vec{O}} \left( G_{B\Vec{N}(d-1)} \left( \Vec{z} \right) \right) 
\end{equation}
and:
\begin{equation}
G_{\Vec{N}(d)} \left( \Vec{z} \right) = G_{\Vec{O}} \left( G_{B\Vec{N}(d-1)} \left( \Vec{z} \right) \right) = G_{\Vec{O}} \left(G_{B\Vec{O}} ^ {d-1} \left( \Vec{z} \right) \right) \ .
\label{eq:GNd}
\end{equation}

We will work with the definitions of $N_{\text{ave}}(d)$ (\ref{av}) and $R_{\text{eff}}$ (\ref{R0}) presented in Section \ref{subsec:neighbor_defs}.
We can now use equation \eqref{eq:GNd} to find the value of $R_{\text{eff}}$ for all sufficiently independent graphs.
We will do this using a Taylor expansion of $ G_{\Vec{N}(d)} $ near $\vec{1}$.
Let $\vec{\varepsilon}$ be a small perturbation for the Taylor series:
\begin{equation}
G_{\Vec{N}(d)} \left( \Vec{1} - \Vec{\varepsilon} \right) = \Vec{1} - M(d) \Vec{\varepsilon} \pm O\left(\| \Vec{\varepsilon} \| ^ 2 \right) \ ,
\end{equation}
where
\begin{equation}
M(d)_{i,j} = \frac{\partial G_{N(d)_i} \left( \Vec{z} \right)}{\partial z_j} \mid_{\Vec{z} = \Vec{1}} \ 
\end{equation}

Consider next the Taylor expansion:
\begin{eqnarray}
G_{\Vec{N}(d)} \left( \Vec{1} - \Vec{\varepsilon} \right) &=& G_{\Vec{O}} \left({G_{B\Vec{O}}} ^ {d-1} \left( \Vec{1} - \Vec{\varepsilon} \right) \right) = \Vec{1} - M_{\Vec{O}} M_{B\Vec{O}} ^ {d-1} \Vec{\varepsilon} \pm O\left(\| \Vec{\varepsilon} \| ^ 2 \right) = \nonumber\\
&=& \Vec{1} - \sum_{\lambda \in eig\left(M_{B\Vec{O}}\right)} \lambda^{d-1} M_{\Vec{O}} 
\Vec{v}_\lambda \langle \Vec{v}_\lambda, \Vec{\varepsilon} \rangle \ ,
\end{eqnarray}
where:
\begin{equation}
\left(M_{\Vec{O}}\right)_{i,j} = \frac{\partial \left(G_{\Vec{O}}\right)_i \left( \Vec{z} \right)}{\partial z_j} \mid_{\Vec{z} = \Vec{1}},~~~~\left(M_{B\Vec{O}}\right)_{i,j} = \frac{\partial \left(G_{B\Vec{O}}\right)_i \left( \Vec{z} \right)}{\partial z_j} \mid_{\Vec{z} = \Vec{1}} \ ,
\end{equation}
and $\Vec{v}_\lambda$ is the eigen-vector of $M_{B\Vec{O}}$ with eigenvalue $\lambda$.

We get:
\begin{equation}
R_{\text{eff}} = \lim_{d\rightarrow \infty} \left( \mathbb{E} \sum_{i,j} \pi_i N(d)_{i,j} \right) ^ {\frac{1}{d}} = 
\lim_{d\rightarrow \infty} \left( \sum_{\lambda \in eig(M_{B\Vec{O}})} c(\lambda) \lambda^d \right) ^ {\frac{1}{d}} = 
\max_{\lambda \in eig\left(M_{B\Vec{O}}\right)} \{|\lambda|\} \ , 
\label{eq:R0}
\end{equation}
where $c(\lambda) = \frac{\langle v_\lambda, \vec{1}\rangle \langle M_{\vec{O}} v_\lambda, \vec{\pi}\rangle}{\lambda} = \Theta(1)$ as a function of $d$.

In Section \ref{subsec:poisson_example} we will derive $G_{\Vec{O}}$ and $G_{B\Vec{O}}$ for Poisson distributions.
In this case:
\begin{equation}
G_{\Vec{O}} \left(\vec{z}\right) = G_{B\Vec{O}} \left(\vec{z}\right) = \exp \left(M \left(\vec{z} - 1\right)\right) \ .
\end{equation}
This gives us $M_{B\Vec{O}} = M$, and $R_{\text{eff}} = \max_{\lambda \in eig\left(M \right)} \{|\lambda|\}$, which is consistent with the results of \cite{Minzer_2021}.

Another commonly used notion of the reproductive index is as a threshold indicator for the spread of a disease: when $R_{\text{eff}} \leq 1$, the large out-component should be negligible and when $R_{\text{eff}} > 1$ the large out-component should be a set fraction of the vertices.
This aspect of the reproductive index was considered for multi-type graphs in \cite{Minzer_2021}.
In Figure \ref{fig:R0_opt} we consider a set of multi-type graphs with values of $R_{\text{eff}}$ near $1$. 
It is clear that in the example computed for this figure, our definition of $R_{\text{eff}}$ also acts as a threshold indicator.
We will comment on this in the discussion section.

\subsubsection{Spanned Sets}
\label{subsec:spanned_sets}
Returning to our main problem, we are interested in the distribution
of the sizes of spanned sets in the graph.
Let $S_i$ and $BS_i$ be the distributions of the sizes of the sets spanned from uniformly random and weighted vertices of type $i$, respectively.
$BS_i$ is well defined, since the distribution of the size of the set spanned from any vertex depends only on its out-degree, and we know how to weight the out-degree distribution by the in-degree due to the sufficient independence of the graph.

Consider the generating functions $G_{\Vec{S}}$ and $G_{B\Vec{S}}$, analogous
to $H_0$ and $H_1$ defined by Newman at al. \cite{PRE_GCC}.

Applying the composition lemma once more, we obtain:
\begin{eqnarray}
G_{\Vec{S}} \left( \Vec{z} \right) &=& G_{\Vec{O}} \left(G_{B\Vec{S}} \left( \Vec{z} \right) \right) \ , \nonumber\\
G_{B\Vec{S}} \left( \Vec{z} \right) &=& G_{B\Vec{O}} \left(G_{B\Vec{S}} \left( \Vec{z} \right) \right) \ .
     \label{GGCD}
\end{eqnarray}
Given $G_{\Vec{O}}$ and   $G_{B\Vec{O}}$ we could in principle 
solve (\ref{GGCD}), first for  $G_{B\Vec{S}}$ using the second
equation and then for $G_{\Vec{S}}$ using the first equation.
However, this is a complicated task and we will follow and generalize the procedure taken in 
\cite{PRE_GCC}.
From (\ref{GGCD}) we have:
\begin{equation}
\Vec{u} \defeq G_{B\Vec{S}} \left( \Vec{1} \right) = G_{B\Vec{O}} \left(G_{B\Vec{S}} \left( \Vec{1} \right) \right) = G_{B\Vec{O}} \left( \Vec{u} \right) \ , 
\end{equation}
and the probability that the connected component originating from a random vertex of type $i$ is not small reads: 
\begin{equation}
p_i = 1 - G_{S_i} \left( \Vec{1} \right) = 1 - G_{O_i} \left(G_{B\Vec{S}} \left( \Vec{1} \right) \right)  = 1 - G_{O_i}(\Vec{u}) \ .
\end{equation}

\subsection{Largest Out-Components}

In the previous subsection we derived an equation that relates the probability that the connected component originating from a random vertex is large to the parameters of the random graph. This is similar to the question of whether or not a Galton-Watson process will terminate.
However, we will be interested in a slightly different question, namely: "what is the size of the largest connected component?".
That is, "what is the probability that a random vertex can be reached by a large portion of the graph?".
The difference between these two questions is not important when dealing
with undirected graphs where a vertex can be reached by the giant connected component (GCC) if and only if the GCC can be reached by it.
When the graph is directed, these questions are interchangeable by reversing the edges of the graph.

Define $I_i$ and $BI_i$ to be the distributions of the in-degrees of a random vertex of type $i$ and the in-degrees weighted by its out-degree, respectively (these are the equivalents of $O$ and $BO$ in the reversed graph).
Since our definition of sufficient independence was symmetric, the reversed graph is also sufficiently independent and these are well defined.
Similarly, we will define $C_i$ and $BC_i$ to be the distributions of the size of the component that spans a uniformly random and a weighted vertex. 
These are the equivalent of $S_i$ and $BS_i$ of the reversed graph, and as in the Galton-Watson analysis we may use their generating functions to find the probabilities that they are finite:

\begin{equation}
\Vec{p} = 1 - G_{\Vec{I}} \left(\Vec{u}\right) \ ,
\end{equation}
where 
\begin{equation}
\Vec{u} = {G}_{B\Vec{C}} \left( \Vec{1} \right) = G_{B\Vec{C}} \left( \Vec{u} \right) \ .
\label{eq:root}
\end{equation}

\section{Analytical Examples of Generating Functions}
\label{sec:analytical_examples}

In this section we will consider various degree distributions, show that they are sufficiently independent, and analytically construct the above generating functions.

\subsection{Poisson Distribution}

\label{subsec:poisson_example}

Similar to the single-type case, the most natural model for multi-type random graphs is the case where each in/out degree is chosen from an independent Poisson distribution.
That is why this model has been the focus of much of the previous work on multi-type graphs \cite{original, book, science, Minzer_2021}.

In the case of a multi-variate Poisson degree distribution, each vertex $u$ of type $i$ has an edge to each vertex $v$ of type $j$ w.p. $p_{i,j}$.
There are $n_j$ vertices of type $j$ and $p_{i,j}$ scales as the inverse of the total population size, thus the distribution of the number of $j$ neighbors of any $i$ vertex is close to a Poisson($n_j p_{i,j}$). In the reversed case, the distribution of the number of $j$ vertices of whom $u$ is a neighbor is distributed according to a Poisson($n_j p_{j,i})$.
Altogether this gives us the generating functions:

    \begin{eqnarray}
G_{O_i} \left( \Vec{z} \right) &=& \prod_{1 \leq j \leq r} G_{O_{i,j}} \left( z_j \right) = \prod_{1 \leq j \leq r} G_{Pois\left(\lambda=p_j n_j\right)} \left( z_j \right) = \prod_{1 \leq j \leq r} \exp \left( p_{i,j} n_j \left(z_j - 1\right) \right) = \nonumber\\
&=& \exp \left( \sum_{1 \leq j \leq r} p_{i,j} n_j \left(z_j - 1\right) \right) \ ,
\end{eqnarray}
and similarly:
\begin{equation}
    \begin{aligned}
G_{I_i} \left( \Vec{z} \right) = \exp \left( \sum_{1 \leq j \leq r} p_{j, i} n_j \left(z_j - 1\right) \right) \ .
    \end{aligned} 
\end{equation}

Furthermore, since the in / out degrees are independent, we have $G_{B\Vec{I}} = G_{\Vec{I}}$.
Therefore: 

\begin{equation}
    \begin{aligned}
G_{\Vec{I}} \left(\Vec{z}\right) = G_{B\Vec{I}} \left(\Vec{z}\right) = \exp \left( P^T diag\left( \Vec{n} \right) \left(\Vec{z}-\Vec{1}\right) \right) = \exp \left( M^T \left(\Vec{z}-\Vec{1}\right) \right) \ .
    \end{aligned}
\end{equation}
From here, we can apply our main theorem to conclude that at the end of the epidemic, the fractions of the different types of the infected population are $\Vec{u}$ where 
\begin{equation}
\Vec{u} = \exp\left(P^T diag\left(\vec{n}\right) \left(\Vec{z}-1\right)\right) \ ,
\end{equation}
similar to equation (2.2) of \cite{original}.

\subsection{Graphs with Independent Degree Distributions}

In general, if the in-out degree distributions to each type are independent, then the generating function is a product of generating functions:
\begin{equation}
    G_{I_j} \left(\Vec{z}\right) =
            \sum_{\alpha_1,...,\alpha_r} p_{\alpha_1, ..., \alpha_r} z_1^{\alpha_1} ... z_r^{\alpha_r} =
            \sum_{\alpha_1,...,\alpha_r} p_{\alpha_1} ...p_{\alpha_r} z_1^{\alpha_1} ... z_r^{\alpha_r} =
            \prod_{1 \leq i \leq r} \sum_{\alpha_i} p_{\alpha_i} z_i ^ {\alpha_i} =
            \prod_{1 \leq i \leq r} G_{I_{i,j}} \left(z_i\right) \ .
\end{equation}

Thus for instance, if we consider a $3$ type graph whose first type has in-degrees that are distributed according to a Poisson distribution with mean $\lambda$, a binomial distribution with parameters $n, p$, and an exponential distribution with parameter $\kappa$, then the multivariate generating function reads:
\begin{equation}
    G_{I_1} \left(\Vec{z}\right) = 
                G_{Poiss(\lambda)} \left(z_1\right) G_{Bin(n,p)} \left(z_2\right) G_{Exp(\kappa)} \left(z_3\right) = 
                e^{\lambda \left(z_1 - 1\right)} \left(1 + pz_2 - p\right) ^ n \frac{1-e^{-1/k}}{1- z_3 e^{-1/k}} \ .
\end{equation}

\subsection{Poisson Matrix with a General Activity Level Distribution}

Another type of generalisation of the Poisson matrix distribution was introduced in \cite{science}.
In this paper the authors partition the population into 6 age cohorts and model the transfer of a disease between these cohorts by using the multi-type Poisson distribution, with a certain interaction matrix $M$.
Furthermore, they assume each age cohort is in itself partitioned into 3 sub-categories based upon their levels of interaction with others.
This leads  to a multi-type problem with $r=3\times6 = 18$ types, thus requiring to solve an equation in $18$ variables.

An easier approach is to consider each age cohort as a single type with a non-Poisson degree distribution.
Specifically, we assume as in \cite{science} that each age cohort is made up of three activity levels: $50\%$ of each age cohort have a normal activity level, $25\%$ have a higher activity level (doubling both their probability of being infected and the expected number of secondary infections they will cause) and $25\%$ have a lower activity level (decreasing both their susceptibility and their infectiousness to half that of a normal activity level).
Considering each age cohort as a single type, it is easy to see that:

\begin{equation}
    \begin{aligned}
G_{\Vec{I}} \left(\Vec{z}\right) = \frac{1}{4} \exp \left(\frac{1}{2} M^T \left(\Vec{z}-\Vec{1}\right) \right) + \frac{1}{2} \exp \left( M^T \left(\Vec{z}-\Vec{1}\right) \right) + \frac{1}{4} \exp \left( 2 M^T \left(\Vec{z}-\Vec{1}\right) \right) \ .
    \end{aligned}
\end{equation}
Each activity level also increases the out-degrees of the vertices, causing a bias proportional to the activity level.
Therefore:
\begin{equation}
    \begin{aligned}
G_{BI} \left(\Vec{z}\right) &= \left(\frac{1}{\frac{1}{8} + \frac{1}{2} + \frac{2}{4}}\right) \left(\frac{1}{8} \exp \left(\frac{1}{2} M^T \left(\Vec{z}-\Vec{1}\right) \right) + \frac{1}{2} \exp \left( M^T \left(\Vec{z}-\Vec{1}\right) \right) + \frac{2}{4} \exp \left( 2 M^T \left(\Vec{z}-\Vec{1}\right) \right)\right) = \\
    &= \frac{1}{9} \exp \left(\frac{1}{2} M^T \left(\Vec{z}-\Vec{1}\right) \right) + \frac{4}{9} \exp \left( M^T \left(\Vec{z}-\Vec{1}\right) \right) + \frac{4}{9} \exp \left( 2 M^T \left(\Vec{z}-\Vec{1}\right) \right) \ .
    \end{aligned}
\end{equation}

In the case considered by \cite{science}, which was chosen as a simple example that can be easily computed, it is still possible to use the naive approach.
However, in larger settings this is likely to prove more difficult.
On the other hand, even in the more general case where the activity level of an individual of each type is distributed according to some distribution $\rho(x)$, one can still use our method to solve for the end of the disease without requiring an optimization over a large number of inputs using the identities:
\begin{equation}
    \begin{aligned}
G_{I} \left(\Vec{z}\right) &= \int_x \rho(x) \exp \left( x M^T \left(\Vec{z}-\Vec{1}\right) \right) dx \ ,
    \end{aligned}
    \label{eq:WP_GI}
\end{equation}
and 
\begin{equation}
    \begin{aligned}
G_{BI} \left(\Vec{z}\right) &= \frac{\int_x x \rho(x) \exp \left( x M^T \left(\Vec{z}-\Vec{1}\right) \right) dx}{\int_x x \rho(x) dx} \ .
    \end{aligned}
    \label{eq:WP_GBI}
\end{equation}
Such graphs are sufficiently independent, since the only dependence between the in-degrees and the out-degrees is through the activity levels which are linearly proportional to any single in / out degree.

\section{Epidemic Spread}

In this section we will use our results to improve and generalise the model proposed by Britton et al. \cite{science} of an epidemic whose spread can be temporarily slowed down.
We will first provide a brief description of this model and the numerical methods used in \cite{science}.

We will then simulate a simpler example, showcasing some of the tools we developed over the last few sections.
Next, we will use these tools to reproduce the results of Britton et al. more efficiently and for more general degree distributions.
Finally, we will compare the results of this model with other models when extended to the the multi-type setting.

\subsection{A Multi-Type Model}

The multi-type model of \cite{science} received much attention and
in this subsection we will explore its subtleties.
A standard analysis of the spread of an epidemic in a single-type graph might follow one of two approaches.
Either the model attempts to find the point at which an uncontrolled epidemic will end its spread by finding the expected size of the GCC or the large out-components of the epidemic's graph \cite{He}, or it attempts to find the herd immunity point at which $R_{\text{eff}} = 1$ \cite{herd1, herd3}.
As was shown in \cite{corona3rd}, finding the point at which $R_{\text{eff}} = 1$ is akin to the outcome of an epidemic with limited intervention.

The model proposed in \cite{science} generalizes this approach to multi-type graphs.
It is assumed that some temporary counter-measures are taken that reduce the infectiousness by a multiplicative factor of $0<\alpha\leq1$.
However, when the counter-measures are lifted, the infectiousness returns to its previous state.
The goal of this type of counter-measures is to reach herd immunity with only a small infected population.

If the value of $\alpha$ were to be set to $1$, the epidemic would go on as usual and while we would reach herd immunity in the first outbreak, we would still see a great many cases even after $R_{\text{eff}} < 1$.
If the value of $\alpha$ were set extremely low, then during the first lockdown we would see no outbreak, but that would not prepare us for the lifting of the restrictions.

By choosing the value of $\alpha$ somewhere in the middle, one can reduce the total infected population, since the first wave will be greatly diminished by the counter-measures and the second wave will be diminished by the resistant population from the first wave.
For single type graphs, it can be easily shown that the optimal value of $\alpha$ is such that after lifting the restrictions, the value of $R_{\text{eff}}$ is exactly $1$ and that the resulting total infected population is equal to the forecast of the "$R_0$ model".
The authors of \cite{science} assume this also holds for the multi-type scenario and search for this value of $\alpha$ in several different scenarios.

\subsection{Large Out-Components of Multi-Type Epidemic Spread}

In this section we will model the spread of an epidemic in a multi-type society.
Details of the computations are provided in Appendix A.
We will first apply the tools to a toy model with $r=2$ types and then we will work with the $r=6$ age cohorts considered in \cite{science}, and consider different activity level distributions.

In order to find a solution to equations (\ref{eq:root}), we need to solve a system of nonlinear equations.
While in general this may be a difficult problem, in our case we found that continuous optimization techniques tend to converge to the solution.
In particular, we use a variant of the gradient descent algorithm detailed in appendix A.
In Figure \ref{fig:R0_opt} we test our gradient descent on a set of multi-type graphs.
We set the basic interaction matrix:
\begin{equation}
M = \begin{pmatrix} 1 & 0.15 \\ 0.05 & 1 \end{pmatrix} \ ,
\label{IM}
\end{equation}
and the distribution of activation levels from \cite{science}, and multiplied the interaction matrix by a different factor each time to obtain different values of the reproductive index per equation (\ref{eq:R0}).
This figure shows both the accuracy of our gradient descent, and the threshold phenomena where $R_{\text{eff}} \leq 1$ corresponds to $\Vec{u} = \Vec{1}$ being the only root (meaning that there are no large out-components) and $R_{\text{eff}} > 1$ corresponds to a second root moving further from $\Vec{1}$ the higher $R_{\text{eff}}$ becomes (which translates to larger out-components).

\begin{figure}[H]
\centering
\includegraphics[width=0.7\columnwidth]{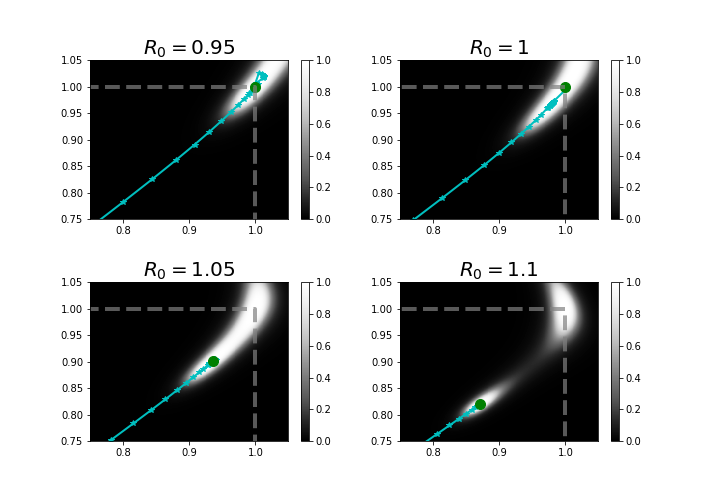} 
\caption{
We use the gradient descent algorithm detailed in appendix A to search for the solution to equation (\ref{eq:root})
with the interaction matrix (\ref{IM}).
We plot a heat-map of $\frac{1}{1+a\left\Vert \Vec{u} - G_{B\Vec{I}} \left( \Vec{u} \right) \right\Vert ^ b}$ where $a,b$ were chosen so that that the zeros of $\Vec{u} - G_{B\Vec{I}} \left( \Vec{u} \right)$ would be more visible.
The axes correspond to the coordinates of $\Vec{u}$, where the closer $\Vec{u}$ is to $\Vec{1}$, the smaller the large out-component is.
The cyan stars indicate the path taken by the gradient descent and the green disk indicates the output of the optimisation.
}
\label{fig:R0_opt}
\end{figure}

In Figure \ref{fig:6_type_model}, we compare the total infected for these age cohorts with several different activity level distributions.
Using equations (\ref{eq:root}),(\ref{eq:WP_GI}),(\ref{eq:WP_GBI}), we compute the expected fractional size of the large out-components of an infectious disease using the population sizes and the interaction matrix from \cite{science}, with varying activity level distributions (the dashed lines).

For each pair of $R_0$ (initial value of $R_{\text{eff}}$) and activity distribution, we performed a population based simulation as detailed in Appendix \ref{appendix:simulation}.
In these simulations we produced a large population of individuals divided amongst a set number of types, with in-out degrees drawn from the appropriate degree distributions (each pair of types have a basic interaction rate which is then multiplied by the activity level of the infecting or infected individual and normalised to obtain the correct $R_0$).

In each simulation we define some small fraction of the population to be infected and the rest to be susceptible.
We then propagate the infection and measure the fraction of the population reached by the disease.
As can be seen in Figure \ref{fig:6_type_model}, the populations based simulations agree almost exactly with the predictions of the multivariate generating function model.
We performed these computations for the case where the activity level distribution is: uniform over the segment $(0,1)$, a half-Gaussian, a Gamma distribution with parameter $\kappa=0.1$ and a generalised Pareto distribution with parameter $\xi=0.25$. In each case we fixed the value of $R_0$ to several values in the range $(1, 5)$ (normalising the interaction matrix to produce any specific $R_0$).
As we can see, the generalised Pareto and the Gamma distribution which are characterised by their higher variance produce a significantly smaller largest out-component than the uniform and half-Gaussian activity level distributions. This fits with the previous results of \cite{corona3rd, herd1, herd3, He, science}.

\begin{figure}[H]
\centering
\includegraphics[width=0.7\columnwidth]{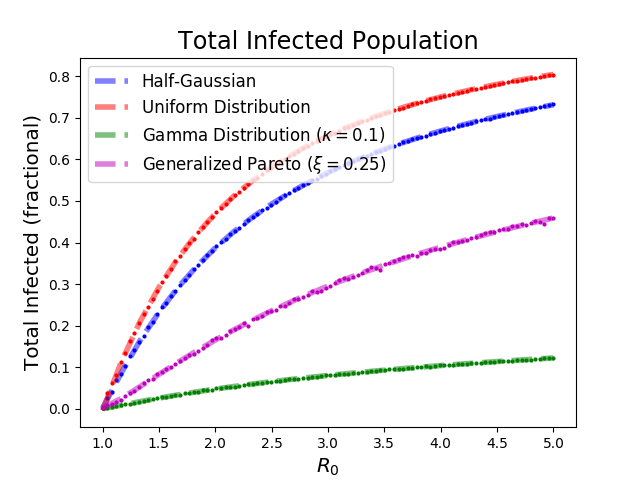} 
\caption{
The expected fractional size of the large out-components of an infectious disease using the population sizes and the interaction matrix from \cite{science}, with varying activity level distributions (the dashed lines) compared to population based simulations for each of these parameter sets (the dots). 
The generalised Pareto and the Gamma distribution which are characterised by their higher variance produce a significantly smaller largest out-component than the uniform and half-Gaussian activity level distributions. 
The results of the population based simulations fit the predictions of the model.
}
\label{fig:6_type_model}
\end{figure}

\section{Limitations}

In this section we will explore some of the limitations of our model.
Its most prominent drawback is that it cannot predict the effects of correlations between the edges of neighboring vertices.
For instance, a contact tracing analysis performed on the outbreak of COVID-19 in Italy showed that the majority of infections were between members of the same household \cite{lavezzo2020suppression}.
Additionally, many physical models (such as the Ising model for instance) only permit interaction between objects that are close to each other (in some geometrical sense).
It is easy to see that information spread through such a graph would not reach exponentially many vertices within a short amount of time.

Our model does not cover this type of graphs.
A concrete example is considered for single-type graphs in Section 6 of \cite{corona3rd}, where it is shown that a population comprised of close knitted families with weak interactions between families would result in a significantly worse final outcome of an epidemic than the one predicted by models that do not take these local interactions into account.
This example holds for the multi-type regime as well.

\subsection{Sufficient Independence}
\label{subsec:lim_suff_ind}

The second limitation of our model is that it works only for graphs with sufficient independence.
To help clarify this notion and show its necessity we will present a simple example without it.
Consider for instance a graph with $r=2$ types, each representing $\pi_1 = \pi_2 = \frac{1}{2}$ of the population.
Let every vertex of type $1$ has exactly $2$ outgoing edges to vertices of type $2$ and no incoming edges.
Finally, let half the vertices of type $2$ have an in-degree of $(4, 0)$ (that is two incoming edges from type $1$ and no incoming edges from type $2$) and an out-degree of $(0, 0)$, and let the rest have an in-degree of $(0, 2)$ and a similar out-degree.
In Figure \ref{fig:dependent_graph} we show a visual representation of such a graph with a population of $N=16$.

In this simple example, it is clear to see that one could simply partition the second type into $2$ very different types and continue with a simple analysis.
However, such partitions are not always possible.
For instance, if instead of having a discrete distribution with only two options for the behaviour of type $2$ vertices, we had set a negative correlation between the first in-degree and the second out-degree, that would have already altered the sufficient independence assumption.

Continuing with the simple example, it is now unclear what the biased distribution of $I_{2,2}$ is.
If we were to give each vertex a weight proportional to its total in-degree, then the vertices with $4$ incoming edges from type $1$ would dominate this distribution and we would estimate that an infected individual of type $2$ is expected, on average, to infect $\frac{2}{3}$ others.
Therefore, given only this piece of information, we would expect that the probability of widespread infection within type $2$ vertices to be very small.
However, this is clearly not true.
In this example, type $1$ vertices have an in-degree of $0$, so they cannot be infected and play no role in the spread of an epidemic.
Some type $2$ vertices take all of their in-degrees from type $1$ vertices and have no out-degrees, so they also take no substantial part in the spread of the infection.
But the rest of the type $2$ vertices act as a small independent system with $R_0 = 2$, and are very likely to experience an outbreak.

In order to create a more general model, one could define $B_i O_{j,k}$ and $B_i I_{j,k}$ to be the distributions of in and out degrees from type $j$ to type $k$, weighted by their $i$th out / in degree.
This would give a slightly more complex additive structure of the form
 \begin{equation}
 B_i N_{j,k} (d+1) = \sum_\ell \sum_{1}^{B_i O_{j, \ell}} B_j N_{\ell, k} (d) \ .
 \end{equation}
While this additive structure could also be used to recreate many of our results in the more general setting, it would require us to produce as input $r^3$ distributions (whereas our construction required an input of only $r^2$ distributions which we set to have a very strong structure anyways).

\begin{figure}[H]
\centering
\includegraphics[width=0.7\columnwidth]{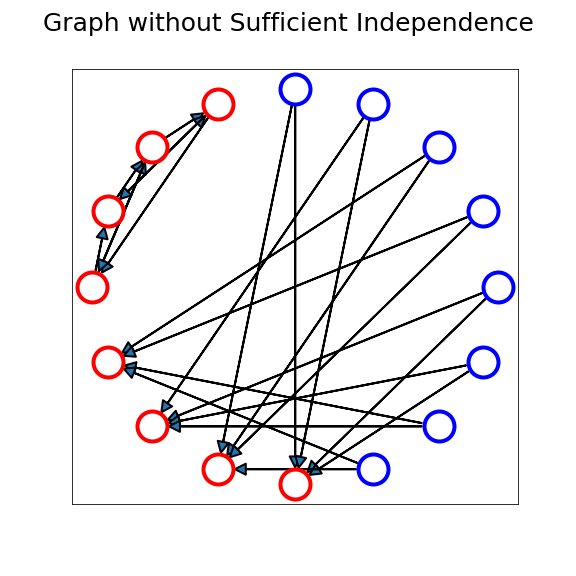} 
\caption{
An example of a graph without sufficient independence. This graph contains $r=2$ types of vertices (the red and blue circles). Each vertex of type $1$ ({\color{blue} blue}) has in-degree $(0, 0)$ and out-degree $(0, 2)$. Half of the vertices of type $2$ ({\color{red} red}) have in-degree $(4, 0)$ and out-degree $(0,0)$ and the other half have in-degree $(0, 2)$ and out-degree $(0, 2)$.
}
\label{fig:dependent_graph}
\end{figure}

\subsection{Unbounded Second Moments}
\label{subsec:stds}

Another limitation of our model is that we require the second moments of the degree distribution to be bounded.
This limitation (which also holds for most single type models) has been explored in a number of previous cases.

We will focus on the example of sexually transmitted diseases \cite{N2, handcock2004likelihood}.
These authors consider the spread of a disease between two populations (types in our nomenclature), where edges are only between the two types and the degrees are assumed to be distributed according to a power-law distributions (with power $\alpha$ commonly assumed to be around $3.2$ \cite{liljeros2001web}).
Furthermore, since such a disease is often modeled by an undirected graph, the biased in/out distributions in our model should be:
\begin{equation}
    \begin{aligned}
        &&\Pr \left[ BO_{i,j} = x \right] \propto x \Pr \left[ O_{i,j} = x \right] \ , \\
        &&\Pr \left[ BI_{i,j} = x \right] \propto x \Pr \left[ I_{i,j} = x \right] \ .
    \end{aligned} 
\end{equation}
Therefore, both $BO_{i,j}$ and $BI_{i,j}$ are distributed according to a power-law with power $\alpha-1$.
Furthermore, since a vertex of any type is only be connected to vertices of one other type (in this model), these distributions are sufficiently independent.
The authors of \cite{N2} show that when the power-law is sub-cubic (i.e. the tail of the degree distribution disappears slower than $\Pr\left[x\right] \propto x^{-3}$), no amount of prevention would stop the spread of such a disease.

In our formulation this is expected when one notices that the second moments of $I_{i,j}$ and $O_{i,j}$ are unbounded.
For instance, if the power-law $\alpha$ of $O_{i,j}$ were less than $3$, then the power-law $\alpha-1$ of $BO_{i,j}$ would be less than $2$, and the average value of $BO_{i,j}$ would be unbounded (which would result in an infinite $R_0$ if our model were relevant to this scenario).
Stricter power-laws (i.e. $\alpha > 3$) would cause $BO_{i,j}$ to have a bounded expectation, resulting in bounded eigenvalues and a bounded value for $R_0$, which fits the results of previous models which show that in these cases the disease can be controlled with counter-measures \cite{N2}.

\section{Discussion}

Previous results have modeled the spread of information in graphs and diseases in populations for either single-type graphs with general degree distributions \cite{PRE_GCC, He, herd1, herd3, herd4, herd5, corona1st, corona3rd}, or for multi-type graphs limited to Poisson \cite{Minzer_2021, original, book} or Poisson-like \cite{science} degree distributions.
In this work we derived a method for computing the size of the large out-component of multi-type graphs with a more complex degree distribution and presented an example of its application to the case of epidemiological research.

In \cite{corona3rd}, the authors compare two types of models for the spread of a disease.
The first is the $R_0$ type model which estimates the percentage of the population infected before herd immunity (defined as the point where $R_{\text{eff}} = 1$).
This model is a good estimate for the percentage of the population infected when the spread of the disease is mitigated by reactive countermeasures that limit the peak infected population.

However, if no countermeasures are taken whatsoever, then at the precise moment when $R_0$ reaches $1$ there will be many infected individuals.
While new infection chains are not expected to grow, those already infected will continue to infect others resulting in a slowly diminishing after-burn effect which can almost double the total infected population.

The second type of models is the GCC type model, which takes into account the after-burn effect and provides a good estimate for the total infected population when no countermeasures are taken.
While both models can be resolved for single-type graphs using the methods shown in \cite{corona3rd}, and for Poisson multi-type graphs as shown in \cite{science}, our solution for the general multi-type scenario resolves only the GCC type model.
We expect that methods similar to those shown in \cite{science} can be used to generalise our results to solve this class of problems as well and leave it for future work.

Another interesting problem is the role of $R_0$:
In our analysis we showed a method by which the concept of a reproductive index (in the sense of the rate by which a disease expands as a function of its generation) can be generalised to classes of random graphs.
Furthermore, we gave a formula for $R_0$ and showed that for Poisson multi-type graphs it is equivalent to the formula given by \cite{Minzer_2021}. 

However, in single-type graphs, the same $R_0$ also plays another role in the study of diseases,
it is a predictor of whether an outbreak will occur or not (there should be an outbreak if and only if $R_0 > 1$).
It can be easily shown that with our definition of $R_0$ for multi-type graphs,  $R_0 > 1$ is a necessary condition for an outbreak.
For instance, we can use the convexity of $G_{BI}$ to obtain the following inequality:

\begin{eqnarray}
    \sum_{\lambda\in eig\left(M^T_{BO}\right)} \vert a_\lambda \vert ^ 2 &= &
    \left\langle \Vec{1} - \Vec{u}, \Vec{1} - \Vec{u} \rangle = \langle \Vec{1} - \Vec{u}, \Vec{1} - G_{BI} \left(\Vec{u}\right) \right\rangle \leq
    \left\langle \Vec{1} - \Vec{u}, \Vec{1} - G_{BI} \left(\Vec{1}\right) + M_{BI} \left(\Vec{1} - \Vec{u}\right) \right\rangle = \nonumber\\
   &=& \left\langle \Vec{1} -  \Vec{u}, M^T_{BO} \left(\Vec{1} - \Vec{u}\right) \right\rangle = 
    \sum_{\lambda\in eig\left(M^T_{BO}\right)} \vert a_\lambda \vert ^ 2 \lambda \ ,
\end{eqnarray}
where
\begin{equation}
\Vec{1} -  \Vec{u} = \sum_{\lambda\in eig\left(M^T_{BO}\right)} a_\lambda \Vec{v_\lambda} 
\end{equation}
is the spectral decomposition of $\Vec{1} -  \Vec{u}$ to the eigenvalues and eigenvectors of $M_{BO}$.
For this inequality to hold, we must have $Re\left(\lambda\right) > 1$ for at least one of the eigenvalues of $M_{BO}$.
However, it still remains to be proven that it is impossible for $R_0$ to be larger than $1$ without having an outbreak.
So far, to the best of our knowledge, this has been proven only for the Poisson distribution in \cite{Minzer_2021}.

\section*{Acknowledgements}

The work is supported in part by the Israeli Science Foundation Center
of Excellence, the European Research Council (ERC) under the European Union’s Horizon 2020 research and innovation program (Grant agreement No. 835152), 
 BSF 2016414, 
the IBM Einstein Fellowship and John and Maureen Hendricks Charitable Foundation
at the Institute for Advanced Study in Princeton.

\bibliography{Corona5th}

\appendix

\section{Methods}

In this appendix we detail the numerical algorithms used to produce the results for Figures \ref{fig:R0_opt} and \ref{fig:6_type_model}.
These included a continuous optimization technique used for finding the solutions to equation (\ref{eq:root}), and a population based simulation whose results were then compared with the output of the model.

\subsection{Gradient Descent}

Gradient descent is a technique used for minimizing functions.
Its base premise is that by moving against the direction of the gradient of the function, one can most rapidly reduce its value.
We utilized this base premise but changed the heuristic by which the distance travelled along this gradient was chosen.
Our goal in this choice was to advance as far as possible against the gradient, without running the risk of overshooting the root.

Our strategy was to frame our root finding problem as the minimization of a function which could be approximated by a quadratic form near its roots.
We did this by choosing to minimise the function:
\begin{equation}
f(\Vec{u}) = \left\Vert \Vec{u} - G_{B\Vec{I}} \left( \Vec{u} \right) \right\Vert ^ 2 \ .
\end{equation}
Let $\Vec{r}$ be the correct solution and define $\Vec{\varepsilon} = \Vec{u} - \Vec{r}$.
Assuming that the Jacobian matrix $J$ of $G_{B\Vec{I}}$ at the root is not equal to the identity, our first order approximation for $f$ when $\varepsilon$ is sufficiently small would be:
\begin{equation}
    \begin{aligned}
    f(\Vec{u}) &=
        \left\Vert \Vec{r}  + \Vec{\varepsilon} - G_{B\Vec{I}} \left( \Vec{r}  + \Vec{\varepsilon} \right) \right\Vert ^ 2 \approx
        \left\Vert \Vec{r}  + \Vec{\varepsilon} - G_{B\Vec{I}} \left( \Vec{r} \right)  - J \Vec{\varepsilon} \right\Vert ^ 2 = \Vec{\varepsilon}^T (I-J)^T (I-J) \Vec{\varepsilon} \ .
    \end{aligned}
    \label{eq:SQR_form}
\end{equation}

We followed these steps to determine how far we went along the gradient at each iteration of the descent
Let $\Vec{u}_n$ be the $n^{th}$ step of the gradient descent.
We define the direction of the step:
\begin{equation}
\hat{d}_n = \frac{\nabla f(u_n)}{\left\Vert \nabla f(u_n)\right\Vert} \ .
\end{equation}
We then consider the simpler function $g_n (x) = f(\Vec{u}_n + x \hat{d}_n)$.
Assuming that $f$ is indeed close to a quadratic form at $u_n$, $g_n$ will also be close to a quadratic form $g_n (x) \approx a x^2 + bx + c$ (for some $a,b,c$).
We use the first and second order derivatives of $g$ at $0$ to estimate its minimum $x_{n} = \frac{-g_n'(0)}{g_n''(0)}$.
However, to be on the safe side we cut the length of the step in half.
Putting it all together, this gives us the following formula for the each step of the gradient descent:
\begin{equation}
    \begin{aligned}
    \Vec{u}_{n+1} = \Vec{u}_n + \frac{1}{2} x_n \hat{d}_n \ .
    \end{aligned}
\end{equation}

\subsection{Population Based Simulation}
\label{appendix:simulation}

We performed population based simulations similar to the ones in \cite{corona3rd}.
In these simulations we modeled a population of size $N=1E6$ individuals with a certain activity level specified for each experiment.
These were then normalised to give the desired $R_0$ using equation (\ref{eq:R0}).
Each individual also carried a state which could be either SUSCEPTIBLE, INFECTED or RECOVERED, and a cohort corresponding to an age group (similar to the age cohorts in \cite{science}).
The percentage of the population in each age cohort and the interaction between age cohorts were set according to the data in \cite{science}.
A small percentage of the population ($100$ individuals chosen at random, weighted according to their activity levels) was then set to INFECTED.

From there we set an iterative process where at each step an INFECTED individual $A$ would set any SUSCEPTIBLE individual $B$ to INFECTED with probability $p = A.infectiousness \times B.susceptibility \times InteractionMatrix[A.cohort, B.cohort]$.
At the end of the iteration we would set $A$ to RECOVERED.
We repeated this process until the entire population was either SUSCEPTIBLE or RECOVERED, and then counted the SUSCEPTIBLE and RECOVERED individuals.

\subsection{Details of the Models}

In all of our models we used the interaction matrix and the relative cohort sizes from \cite{science}.
In particular, these age cohorts which correspond to the age groups 0-5, 6-12, 13-19, 20-39, 40-59 and 60+ were assumed to represent $7.25\%$, $8.66\%$, $11.24\%$, $33.23\%$, $22.67\%$ and $16.95\%$ of the population respectively.

The basic interaction matrix used was: 
\[
\left(\begin{matrix}
2.2257 & 0.4136 & 0.2342 & 0.4539 & 0.2085 & 0.1506 \\
0.4139 & 3.614 & 0.4251 & 0.4587 & 0.2712 & 0.1514 \\
0.2342 & 0.4257 & 2.9514 & 0.6682 & 0.4936 & 0.1972 \\
0.4539 & 0.4592 & 0.6676 & 0.9958 & 0.651 & 0.33 \\
0.2088 & 0.2706 & 0.4942 & 0.6508 & 0.8066 & 0.4341 \\
0.1507 & 0.152 & 0.1968 & 0.3303 & 0.4344 & 0.7136 \\
\end{matrix}\right) \ .
\]

We considered 4 basic distributions of activity levels:

\begin{itemize}
    \item {\em Gamma Distribution} with the PDF $f(x) = \frac{1}{\Gamma(k) \theta^k} x^{k-1} e^{-\frac{x}{\theta}}$, where $k$ was set to $0.1$ (as in \cite{corona3rd}) and $\theta$ was scaled for the desired $R_0$.
    \item {\em Generalised Pareto Distribution} with the PDF $f(x) = \frac{1}{\sigma} \left(1 + \zeta \frac{x-\mu}{\sigma}\right) ^ {-(1/\zeta+1)}$, where $\mu$ and $\zeta$ were set to $0$ and $0.25$ respectively and $\sigma$ was scaled for the desired $R_0$.
    \item {\em Half Normal Distribution} with the PDF $f(x) = \sqrt{\frac{2}{\pi}} \frac{1}{\sigma} \exp\left(-\frac{x^2}{2\sigma}\right)$ (for all $x>0$), where $\sigma$ was scaled for the desired $R_0$.
    \item {\em Uniform Distribution} with the PDF $f(x) = \frac{1}{L}$ for all $x$ in the segment $\left[0,L\right]$ and $0$ otherwise, where $L$ was scaled for the desired $R_0$.
\end{itemize}

\end{document}